\documentclass[11pt]{article}
\usepackage{CJKutf8}
\usepackage{graphicx,wrapfig}
\DeclareGraphicsExtensions{.pdf,.png,.jpg}

\usepackage[utf8]{inputenc}
\usepackage[T1]{fontenc}
\usepackage{bm}
\usepackage{amsmath,amsfonts,amssymb,multirow}
\usepackage{amsthm}
\usepackage{multicol}
\usepackage{appendix}
\usepackage{cite}
\usepackage[margin=1in]{geometry}
\usepackage{fullpage} 
\usepackage{graphics}
\usepackage{float}
\usepackage[dvipsnames,usenames]{color}
\usepackage[colorlinks=true,urlcolor=Blue,citecolor=Green,linkcolor=BrickRed]{hyperref}
\usepackage[usenames,dvipsnames]{xcolor}
\usepackage{enumerate}
\usepackage{caption}
\usepackage{subcaption}

\usepackage{algorithm}
\usepackage[noend]{algorithmic}

\usepackage{paralist}
\usepackage{todonotes}

\usepackage[charter]{mathdesign}
\usepackage{berasans,beramono}
\usepackage{microtype}

\usepackage{mathtools}

\newtheorem{lemma}{Lemma}
\newtheorem{theorem}{Theorem}
\newtheorem{corollary}{Corollary}

\usepackage{arydshln}

\newcommand\floor[1]{\lfloor#1\rfloor}
\newcommand\ceil[1]{\lceil#1\rceil}

\usepackage{todonotes}

\def\etal{\emph{et~al.}}

\newcommand{\bd}{\partial\!}
\def\slack{\operatorname{\mathit{slack}}}
\def\EMPH#1{\textbf{\emph{\boldmath #1}}}

\makeatletter

\begin{document}

\title{Near-Optimal Distance Emulator for Planar Graphs}

\author{Hsien-Chih Chang\thanks{University of Illinois at Urbana-Champaign, USA. Partially supported by NSF grant CCF-1408763.
    }
    \and
  Pawe\l{} Gawrychowski\thanks{University of Wrocław, Poland. 
  }
\and
 Shay Mozes\thanks{IDC Herzliya, Israel. Partially supported by  the Israel Science Foundation grants 794/13 and 592/17.
 }
\and
  Oren Weimann\thanks{University of Haifa, Israel. Partially supported by  the Israel Science Foundation grants 794/13 and 592/17.}
}

\date{}
\maketitle

\begin{abstract}
Given a graph $G$ and a set of terminals $T$, a \emph{distance emulator} of $G$ is another graph $H$ (not necessarily a subgraph of $G$) containing $T$, such that all the pairwise distances in $G$ between vertices of $T$ are preserved in $H$.
An important open question is to find the smallest possible distance emulator.

We prove that, given any subset of $k$ terminals in an $n$-vertex undirected unweighted planar graph, we can construct in $\tilde O(n)$ time a distance emulator of size $\tilde O(\min(k^2,\sqrt{k\cdot n}))$. This is optimal up to logarithmic factors. 
The existence of such distance emulator provides a straightforward framework to solve distance-related problems on planar graphs: Replace the input graph with the distance emulator, and apply whatever algorithm available to the resulting emulator. 
In particular, our result implies that, on any unweighted undirected planar graph, one can compute all-pairs shortest path distances among $k$ terminals in $\tilde O(n)$ time when $k=O(n^{1/3})$.   
\end{abstract}

\section{Introduction}

The planar graph metric is one of the most well-studied metrics in graph algorithms, optimizations, and computer science in general. The planar graph \emph{metric compression} problem is to compactly represent the distances among a subset $T$ of $k$ vertices, called \EMPH{terminals}, in an $n$-vertex planar graph $G$. 
Without compression (and when $G$ is unweighted) these distances can be na\"ively represented with $O(\min(k^2  \log n,n))$ bits by either explicitly storing the $k \times k$ distances 
or alternatively by storing the entire graph (na\"ively, this takes $O(n \log n)$ bits, but can be done with $O(n)$ bits \cite{Turan84,MunroR97,ChiangLL01,BF10}). 

A natural way to compress $G$ is to replace it by another graph $H$ that contains $T$ as vertices, and the distances between vertices in $T$ are preserved in $H$. In other words, $d_G(x,y) = d_H(x,y)$ holds for every pair of vertices $x$ and $y$ in $T$. Such graph $H$ is called a \EMPH{distance emulator} of $G$ with respect to $T$ (or a \EMPH{distance preserver} in the case where $H$ is required to be a subgraph of $G$). 
Distance emulators are appealing algorithmically, since we can readily feed them into our usual graph algorithms. They have been studied as compact representations of graphs \cite{bce-sdpas-05,ce-sspdp-06,ChlamtacDKL17,BodwinW16,Bodwin17,Alon02}, and used as fundamental building blocks in algorithms and data structures for distance-related problems \cite{bce-sdpas-05,dhz-apasp-00,ElkinP16,AbboudB16,AbboudB17,AbboudBP17,HuangP18}.  Similar concepts like additive and multiplicative \emph{spanners} \cite{dhz-apasp-00,g-sptmh-01,bkmp-ncspa-05,tz-sesde-06,w-lbase-06,bg-sprgm-08}
and \emph{distance labelings} \cite{gp-egsc-72,w-pscc-83,m-mug-65,knr-irg-92,gppr-dlg-04,adkp-sdl-16,gku-ssdlu-16,aghp-sfsld-16}
are popular topics with abundant results.

In planar graphs, an extensive study has been done for the case where the emulator $H$ is required to be a minor of $G$ \cite{g-sptmh-01,CXKR06,EEST08,bg-sprgm-08,EGK+14,CGH16,GR16,GHP17,KR17arxiv}. Restricting $H$ to be a minor of $G$ has proven useful when $H$ is only required to \emph{approximate} the distances between vertices in $T$, say, up to a $(1 + \varepsilon)$ multiplicative error. For exact distances however, Krauthgamer, Nguyen, and Zondiner \cite{KNZ14} have shown a lower bound of $\Omega(k^2)$ on the size of $H$, even when $G$ is an unweighted grid graph and   $H$ is a (possibly weighted) minor emulator.

In general, an emulator $H$ does not have to be a subgraph or a minor of $G$. In fact, $H$ can be non-planar even when $G$ is planar. Even in this setting, for \emph{weighted} planar graphs,  we cannot beat the na\"ive bound because one can encode an arbitrary $k \times k$ binary matrix using subset distances among $2k$ vertices in a weighted planar graph \cite{gppr-dlg-04,agmw-nocpg-17}. Since we cannot compress an arbitrary binary $k \times k$ matrix into less than $k^2$ bits, we again have an $\Omega(k^2)$ lower bound.

What about \emph{unweighted} planar graphs? Various distance-related problems in unweighted graphs enjoy algorithms and techniques \cite{dhz-apasp-00,kk-spqpg-03,bkmp-ncspa-05,tz-sesde-06,w-lbase-06,w-widpg-09,WYgirth10,c-cgpgl-13,w-ctdqp-13,ek-ltamf-13} which outperform their weighted counterparts.  Indeed, for the metric compression problem in unweighted planar graphs, Abboud, Gawrychowski, Mozes, and Weimann \cite{agmw-nocpg-17} have very recently shown that we can in fact beat the $O(k^2)$ bound.  They showed that the distances between any $k$ terminals can be encoded using $ \tilde O(\min(k^2,\sqrt{k\cdot n}))$ bits\footnote{The $\tilde{O}(\cdot)$ notation hides polylogarithmic factors in $n$.}. 
The encoding of Abboud \etal\ is optimal up to logarithmic factors. However, it is not an emulator! Goranci, Henzinger, and Peng \cite{GHP17} raised the question of whether such encoding can be achieved by an emulator. We answer this question in the affirmative.

\paragraph*{Our results.}
We show that the encoding of Abboud \etal\ can be turned into an emulator, with no asymptotic overhead. Namely, we prove that any unweighted planar graph has a near-optimal size distance emulator.

\begin{theorem}
\label{Th:emulator-subset}
Let $G$ be an $n$-vertex undirected unweighted planar graph, and let $T$ be the set of $k$ terminals. A directed weighted graph $H$ of size $O(\min(k^2,\sqrt{k\cdot n}\log^3n))$ can be constructed in $O(n\log^4 n)$ time as a distance emulator of $G$ with respect to $T$; all edge weights of $H$ are non-negative integers bounded by $n$. 
\end{theorem}

Our theorem provides a practical framework for solving distance-related problems on planar graphs:  Replace the input graph $G$ (or proper subgraphs of the right choice) with the corresponding distance emulator $H$, and invoke whatever algorithm available on $H$ rather than on $G$. 
One concrete example of this is the computation of all-pairs shortest paths among a subset $T$ of $k$ vertices.
The algorithm by Cabello \cite{c-mdpg-12} can compute these $\Theta(k^2)$ distances in $\tilde{O}(n^{4/3} + n^{2/3} k^{4/3})$ time.   
Mozes and Sommer \cite{ms-edopg-12} improved the running time to $\tilde{O}(n^{2/3} k^{4/3})$ 
when 
$k = \Omega(n^{1/4})$. 
Our emulator immediately implies a further improvement of the running time to $\tilde{O}(n+n^{1/2} k^{3/2})$ 
using an extremely simple algorithm: Replace $G$ by the emulator $H$ and run Dijkstra's single-source shortest-paths algorithm on $H$ from every vertex of $T$.
This improves on previous results by a polynomial factor for essentially the entire range of $k$.

\begin{corollary}
\label{cor}
Let $G$ be an undirected unweighted planar graph, and let $T$ be a subset of $k$ terminals in $G$.  One can compute the shortest path distances between all pairs of vertices in $T$ in $O(n\log^4 n + n^{1/2} k^{3/2} \log^3 n \log\log n)$ time.
\end{corollary}

\paragraph*{Techniques.}
The main novel idea is a graphical encoding of unit-Monge matrices. 
Our emulator is obtained by plugging this graphical encoding into the 
encoding scheme of Abboud \etal~\cite{agmw-nocpg-17}. Their encoding consists of two parts. 
First, they explicitly store a set of distances between some pairs of vertices of $G$ (not necessarily of $T$). 
It is trivial to turn this part into an emulator by adding a single weighted edge between every such pair. 
Second, they implicitly store all-pairs distances in a set containing pairs of the form $(X,Y)$ where $X$ is a \emph{prefix} of a cyclic walk around a single face and $Y$ is a \emph{subset} of vertices of a cyclic walk around a (possibly different) single face. 
For each such pair $(X,Y)$, the $X$-to-$Y$ distances are efficiently encoded using only $O((|X|+|Y|) \cdot \log n)$ bits (rather than the na\"ive $O(|X|\cdot|Y| \cdot \log n)$ bits).  
This encoding follows from the fact that the $X\times Y$ distance matrix $M$ is a \emph{triangular unit-Monge} matrix.
\vspace{0.037in} 

\noindent A matrix is \EMPH{Monge} if for any $i,j$ we have that 
\[
M[i+1,j] - M[i,j] ~\le~ M[i+1,j+1] - M[i,j+1].
\] 
A Monge matrix is \EMPH{unit-Monge} if for any $i,j$ we also have that 
\[
M[i+1,j]-M[i,j] ~\in~ \{-1,0,1\}.
\]
A (unit-)Monge matrix is \EMPH{triangular} if the above conditions are only required to hold for $i\neq j$.  In other words, when all four entries $M[i,j], M[i+1,j], M[i,j+1], M[i+1,j+1]$ belong to the upper triangle of $M$ or to the lower triangle of $M$. 

The Monge property has been heavily utilized in numerous algorithms for distance related problems in planar graphs \cite{fr-pgnwe-06,c-sadsp-17,ms-edopg-12,insw-iamcm-11,TALG10SP,MSMS,GomoryHu,NussbaumOracle}  
in computational geometry \cite{kmns-smqmm-12,SubmatrixMonge,Kaplan2,Kaplan3,SMAWK,KK89}, 
and in pattern matching \cite{Schmidt1998,CrochemoreLandauZiv,Algorithmica13,TiskinSODA}. However, prior to the work of Abboud \etal\ \cite{agmw-nocpg-17}, the stronger \emph{unit-Monge} property has only been exploited in pattern matching \cite{TiskinSODA,Tiskin}. 

The encoding of the $X\times Y$ unit-Monge matrix $M$ is  immediate: For any $i$, the sequence of differences $M[i+1,j]-M[i,j]$ is nondecreasing and contains only values from $\{-1,0,1\}$, so they can be encoded by storing the positions of the first 0 and the first 1. Storing these positions for every $i$ takes $ O (|X| \cdot \log n)$ bits. To encode $M$, we additionally store $M[0,j]$ for every $j$ using $ O(|Y|\cdot \log n)$ bits. This encoding, however, is clearly not an emulator. Our main technical contribution is in showing that it can be turned into an emulator with no asymptotic overhead. Namely, in Section~\ref{S:monge-emulator} we prove the following lemma.

\begin{lemma}
\label{L:emulator}
Given an $n \times n$ unit-Monge or triangular unit-Monge matrix $M$, one can construct in $O(n \log n)$ time a directed edge-weighted graph (emulator) $H$ with $O(n \log n)$ vertices and edges such that the distance in $H$ between vertex $i$ and vertex $j$ is exactly $M[i,j]$.
\end{lemma}

In Section~\ref{S:construction} we describe how to plug the emulator of  Lemma~\ref{L:emulator} into the encoding scheme of Abboud \etal \cite{agmw-nocpg-17}. 
In their paper, the size of the encoding was the only concern and the construction time was not analyzed. In our case, however, the construction time is crucial for the application in Corollary~\ref{cor}.  We achieve an $O(n \log^4 n)$ time construction by making a small modification to their encoding, based on the technique of heavy path decomposition. 

Another difference is that in their construction once the encoding is computed, single-source shortest-paths can be computed on the encoding itself using an algorithm by Fakcharoenphol and Rao~\cite{fr-pgnwe-06}. More precisely (see~\cite[Section 5]{agmw-nocpg-17}), using a compressed representation of unit-Monge matrices, the total size of their encoding is $s = O(\sqrt{k\cdot n} \log^2 n)$ words, running the Fakcharoenphol and Rao algorithm requires accessing $O(s \log^2 s)$ elements of the unit-Monge matrices, and accessing each such element from the compressed representation requires $O(\log n/\log\log n)$ time (via a range dominance query, see Section~\ref{S:monge-emulator})\footnote{The $O(\log n/\log\log n)$ factor was overlooked in~\cite{agmw-nocpg-17}. It can be obtained by representing the unit-Monge matrix with a 
range dominance data structure, such as the one in~\cite{JaJaMS04}.}. Overall, using our $O(n\log^4 n)$ construction together with the above single-source shortest-paths computations from each terminal would result in a time bound of $O(n\log^4 n + n^{1/2} k^{3/2} \log^5 n/\log\log n)$ in Corollary~\ref{cor}.
Our approach on the other hand, runs Dijkstra's classical algorithm on the emulator $H$, leading to a faster (by a $(\log n/ \log\log n)^2$ factor) time bound for Corollary~\ref{cor}. More precisely, our emulator $H$  is of size $O(\sqrt{k\cdot n}\log^3 n)$, the edge-weights in $H$ are encoded explicitly (i.e. there is no need for a range dominance query), and the edge-weights are integers in $[1 \,..\, n]$ so Dijkstra's algorithm can use an $O(\log\log n)$ time heap (using such fast heap is not possible for Fakcharoenphol and Rao's algorithm). Running Dijkstra therefore takes $O(\sqrt{k\cdot n}\log^3 n \log\log n)$ time leading to the bound in Corollary~\ref{cor}.

\section{Distance Emulators for Unit-Monge Matrices}
\label{S:monge-emulator}

\subsection{The bipartite case}
We next describe an $O(n\log n)$ space and time construction of a distance emulator for \emph{square} $n\times n$ unit-Monge matrices.  
The construction can be trivially extended to \emph{rectangular} $n\times m$ unit-Monge matrices (in $O(\max(n,m) \log (\max(n,m)))$ time and space) by duplicating the last row or column until the matrix becomes square.

We begin by establishing a relation between unit-Monge matrices and right substochastic binary matrices.
A binary matrix $P$ is \EMPH{right substochastic} if every row of $P$ contains at most one nonzero entry.  
The following lemma can be established similarly to Abboud \etal~{\cite[Lemma~6.1]{agmw-nocpg-17}}.%
\footnote{In Abboud \etal\ \cite[Section 6]{agmw-nocpg-17} a unit-Monge matrix is defined to be one where every two adjacent elements differ by at most 1, whereas here we only assume this for elements that are adjacent vertically (that is, in the same column and between adjacent rows).} 

\begin{lemma}[Abboud \etal~{\cite[Lemma~6.1]{agmw-nocpg-17}}]
\label{L:substochastic}
For any $n \times m$ unit-Monge matrix $M$ there is a right substochastic $2(n-1) \times (m-1)$ binary matrix $P$ such that for all $x$ and $y$,
\[
M[x,y] ~=~ U[x] + V[y] + \sum_{\substack{i\geq 2x-1 \\ j\geq y}} P[i,j]
\]
where $U$ is a vector of length $n$ and $V$ is a vector of length $m$. 
\end{lemma}

\begin{proof}
We first define an $(n-1)\times (m-1)$ matrix $P'$ as follows:
\[
P'[i,j] = M[i+1,j+1] + M[i,j] - M[i,j+1] - M[i+1,j].
\]
By Monge, clearly $P'[i,j] \geq 0$, and by unit $P'[i,j]\leq 2$. In fact, unit also implies that
the sum in every row of $P'$ is at most 2. To see this, consider
$P'[i,1]+P'[i,2]+\ldots+P'[i,n]$. After telescoping, this is $P'[i,1]-P[i+1,1]+P[i+1,m]-P[i,m]$,
so by unit at most 2 as claimed. Then, consider $\sum_{i' \geq i,j'\geq j} P'[i',j']$. After substituting
the definition of $P'$ and telescoping, this becomes $M[i,j]+M[n,m]-M[i,m]-M[n,j]$.
Hence, if we define $U[i]=M[i,m]-M[n,m]$ and $V[j]=M[n,j]$ it holds that
$M[i,j] = U[i]+V[j] + \sum_{i' \geq i, j' \geq j} P'[i',j']$. Finally, we create a $2(n-1) \times (n-1)$
matrix $P$, where every $2\times 1$ block corresponds to a single $P'[i,j]$, that is,
the sum of values in the block is equal to $P'[i,j]$. It is always possible to define
$P$ so that it is a permutation matrix. To see this, consider a row of $P'$. The values
there sum up to at most 2, say $P'[i,j]=P'[i,j']=1$ for some $j<j'$. Then, $P'[i,j]$ should
correspond to a 1 in the first row of its block and $P'[i,j']$ to a 1 in the second row of
its block. If $j=j'$ then in the corresponding block we create two 1s, one per row.
\end{proof}

Our goal is to construct a small emulator (in terms of number of vertices and edges) for $M$.
By Lemma~\ref{L:substochastic}, it suffices to construct a graph $H$ satisfying 
\[
d_H[r[x],c[y]] ~=~ \sum_{\substack{i\geq x \\ j\geq y}}P[i,j],
\]
where $x$ and $y$ range over the rows and columns of $P$, $r[x]$ is a vertex corresponding to row $x$ of $P$, $c[y]$ is a vertex corresponding to column $y$ of $P$, and the $r[x]$-to-$c[y]$ distance in $H$ equals the number of 1s in $P$ dominated by entry $P[x,y]$ (a 2-dimensional range dominance counting).
We assume that $P$ is given as input, which is represented as a vector specifying, for each row $i$ of $P$, the index of the (at most) single one entry in that row.
We refer to such a graph $H$ as an \EMPH{emulator} for the right substochastic matrix $P$.
To convert an emulator for $P$ into an emulator for $M$, add a new vertex $r_0[x]$ for each $1\leq x < n$, and connect it with an edge of weight $U[x]$ to $r[2x-1]$.  Similarly, for each $1\leq y < m$, connect $c[y]$ to a new vertex $c_0[y]$ with an edge of weight $V[y]$. 
In addition, add a new vertex $c_0[m]$ and connect each $r_0[x]$ to $c_0[m]$ by an edge of weight $U[x]+V[m]$; and add a new vertex $r_0[n]$ and connect $r_0[n]$ to each $c_0[y]$ by an edge of weight $U[n]+V[y]$.
By Lemma~\ref{L:substochastic}, the $r_0[x]$-to-$c_0[y]$ distance equals $M[x,y]$.
We therefore focus for the remainder of this section on constructing an emulator for the right substochastic matrix $P$.

\begin{figure}[t]
\centering
\includegraphics[width=\linewidth]{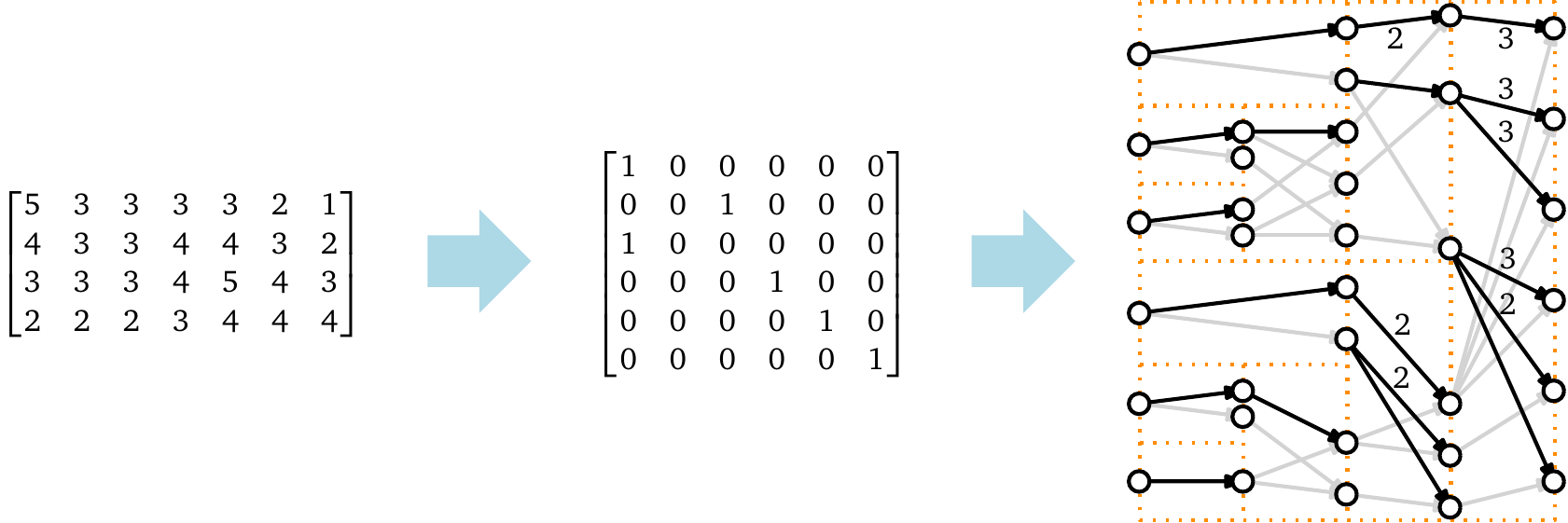}
\caption{Unit-Monge distance matrix $M$, its corresponding row substochastic matrix $P$, and the emulator~$H$ for the matrix $P$.  
In this example the two vectors given by Lemma~\ref{L:substochastic} are $U = [{-3} \>\> {-2} \>\> {-1} \>\> 0]$ and $V = [2 \>\> 2 \>\> 2 \>\> 3 \>\> 4 \>\> 4 \>\> 4]$.
The leftmost vertices in $H$ are row vertices $r[x]$, and the rightmost vertices in $H$ are column vertices $c[y]$.  
All gray directed edges without labels have weight $0$, and all black directed edges without labels have weight $1$.  
For instance, the distance from $r[2]$ to $c[1]$ is equal to the number of $1$s in $P$ dominated by the $(2,1)$-entry, which is $5$.
}
\label{F:emulator}
\end{figure}

\paragraph{Recursive construction.}
We now describe how to construct an emulator for an $n \times \alpha n$ right substochastic matrix $P$ recursively for arbitrary constant $\alpha \ge 0$.

When $P$ has only a single row, 
we create a row vertex $r$ for the single row and a column vertex $c[y]$ for each column $y$.  Connect $r$ to $c[y]$ with a directed edge of weight $\smash\sum_{j \ge y} P[1,j]$.  

Otherwise, assume $P$ has more than one single row.
Divide $P$ into the top $\floor{n/2} \times \alpha n$ submatrix \EMPH{$P_\uparrow$} and the bottom $\ceil{n/2} \times \alpha n$ submatrix \EMPH{$P_\downarrow$}.
Since $P_\uparrow$ is right substochastic, $P_\uparrow$ has at most $\floor{n/2}$ columns that are not entirely zero; similarly $P_\downarrow$ has at most $\ceil{n/2}$ non-zero columns.
Let \EMPH{$P'$} and \EMPH{$P''$} be the submatrices of $P_\uparrow$ and $P_\downarrow$ respectively, induced by their non-zero columns as well as their last column (if it's a zero vector); denote the number of columns of $P'$ and $P''$ as $n'$ and $n''$, respectively.
Observe that both $P'$ and $P''$ are still right substochastic matrices.
Let $y'_1,\dots,y'_{n'}$ be the indices in $P_\uparrow$ (and thus in $P$) corresponding to the columns of $P'$, and let $y''_1,\dots,y''_{n''}$ be the indices in $P_\downarrow$ corresponding to the columns of $P''$.  
Observe that for any row $x$, the sum 
\[
\sum_{\substack{i\geq x \\ j\geq y}} P_\uparrow[i,j]
\]
is the same for all $y \in (y'_{\ell-1} \,..\, y'_{\ell}]$ for any fixed $\ell$.%
\footnote{For simplicity, one assumes $y'_0 = y''_0 = -\infty$.}  
This is because all the columns of $P_\uparrow$ in the range $(y'_{\ell-1} \,..\, y'_{\ell} - 1]$ are all zero.  A similar observation holds for the matrix $P_\downarrow$.  For each $y$, we denote by \EMPH{$y'_\to$} the (unique) smallest index $y'_\ell \in \{y'_1,\dots,y'_{n'}\}$ no smaller than $y$, and denote by \EMPH{$y''_\to$} the smallest index $y''_\ell \in \{y''_1,\dots,y''_{n''}\}$ no smaller than $y$. 

\begin{figure}[t]
\centering
\includegraphics[width=.7\linewidth]{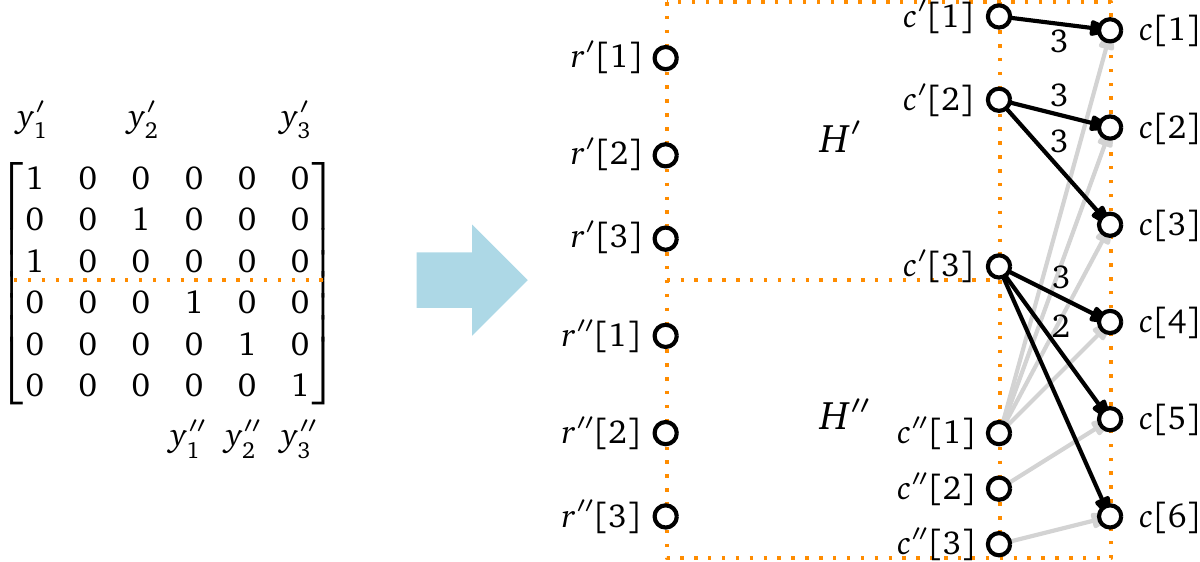}
\caption{The top-level construction of the emulator.  As an example, the edge from $c'[2]$ to $c[2]$ exists because $y'_\to = 2$ when $y=2$, and the weight on the edge is equal to $\sum_{{i > 3,j\geq 2}} P[i,j] = 3$.}
\label{F:recursive}
\end{figure}

We recursively compute emulators $H'$ and $H''$ for $P'$ and $P''$ respectively.  
Denote the row and column vertices of $H'$ as $r'[x]$ and $c'[y]$ where $x$ and $y$ range over the rows and columns of $P'$, respectively; similarly denote the row and column vertices of $H''$ as $r''[x]$ and $c''[y]$ respectively.
Take the row vertices of $H'$ and $H''$ as the row vertices of $H$, respecting the indices of $P$.
For each column $y$ of $P$ we add a column vertex $c[y]$ to $H$;
connect vertex $c''[y''_\to]$ of $H''$ to $c[y]$ with an edge of weight zero; and
connect vertex $c'[y'_\to]$ of $H'$ to $c[y]$ with an edge of weight 
\[
\sum_{\substack{i > \floor{n/2} \\ j\geq y}} P[i,j]. 
\]
Thus, since for each pair of indices $(x,y)$ of $P$ one has
\[
d_{H'}[r[x],c'[y]] ~=~ \sum_{\substack{x \le i \le \floor{n/2} \\ j\geq y}} P_\uparrow[i,j] \quad \text{and} \quad d_{H''}[r[x],c''[y]] ~=~ \sum_{\substack{i\geq x \\ j\geq y}} P_\downarrow[i,j],
\]
we have that 
\[
d_H[r[x],c[y]] ~=~ \sum_{\substack{i\geq x \\ j\geq y}} P[i,j]
\]
as desired. 
This completes the description of the construction of the emulator for square right substochastic matrices. Observe that the obtained emulator is a directed acyclic graph (dag) with a single path from each row vertex to each column vertex.  

\paragraph{Size analysis.}
We next analyze the size of the emulator.  Row vertices are created only at the base of the recursion and there are $n$ of them in total.  The number of edges and non-row vertices $N(n)$ satisfies the recurrence $N(n) \leq 2N(n/2) + O(\alpha n)$, so there are overall $O(n \log n)$ vertices and edges in the emulator when the number of columns is linear to the number of rows $n$.

\paragraph{Time analysis.}
As for the construction time, the only operation in the recursive procedure that does not clearly require constant time is the computation of the edge weights. However, it is not hard to see that computing all the edge weights of the form 
\[
w_y ~\coloneqq~ \sum_{\substack{i > \floor{n/2} \\ j\geq y}} P[i,j]
\]
in a single recursive call can be done in linear time.  This is because $w_y$ is precisely the weight of the unique path in the emulator $H''$ from the first row vertex to the $y$'th column vertex.  
We can therefore simply maintain (with no asymptotic overhead), for every recursively built emulator, the distances from its first row vertex to all its column vertices.

\paragraph{Making weights non-negative.}
The weight of every edge in our emulator for right substochastic matrix $P$ is non-negative, but the vectors $U$ and $V$ might contain negative entries; as a result the obtained emulator for $M$ might contain negative edges even though every entry in $M$ is non-negative.  This can be fixed by a standard reweighting strategy using a \emph{price function}. 
For the sake of computing the price function $\phi(\cdot)$, add a super-source $s$ and connect $s$ to each row vertex $r[x]$ of the emulator of $M$ by an edge of weight zero.  Compute the single-source shortest path tree rooted at $s$; since our emulator for $P$ is a dag, this can be done in time linear in the size of the emulator.  Take the shortest path distances $d(s,\cdot)$ as the price function $\phi(\cdot)$, and let the \emph{reduced weight} $w_\phi(uv)$ of edge $uv$ be $w(uv) + \phi(u) - \phi(v)$.  Finally increase the weight of every incoming edge to $c[y]$ by the amount $\phi(y)$. 
Now all modified weights are non-negative, and all shortest path distances from $r[x]$ to $c[y]$ are preserved because $\phi(r[x]) = 0$ holds for all row vertices $r[x]$.  

\subsection{The non-bipartite case}
Recall that the matrices that capture pairwise distances among consecutive vertices on the boundary of a single face in a planar graph are not (unit-)Monge. In Abboud \etal\ \cite{agmw-nocpg-17} this is handled by decomposing such a matrix recursively into (unit-)Monge submatrices, which incurs a logarithmic overhead. We show how to avoid this logarithmic overhead by constructing emulators for triangular (unit-)Monge matrices. 
Let $v_1,v_2, \dots, v_n$ be vertices on the boundary of a single face. Let $M$ be an $n \times n$ matrix such that $M[i,j]$ is the distance from $v_i$ to $v_j$. For any quadruple $(i,i',j,j')$ satisfying $i < i' < j < j'$, the shortest $v_i$-to-$v_j$ path must intersect the shortest $v_{i'}$-to-$v_{j'}$ path.  Therefore, $M[i,j] + M[i',j'] \geq M[i,j'] + M[i',j]$.  Unfortunately, this does not necessarily hold for an arbitrary quadruple. 
We can, however, define two $n \times n$ unit-Monge matrices $M_l$ and $M_u$, such that 
\begin{align*}
M_l[i,j] = M[i,j] & \qquad \text{for all $n \geq i>j \geq 1$, and} \\
M_u[i,j] = M[i,j] & \qquad \text{for all $1\leq i<j\leq n$.}
\end{align*}
The other elements of $M_l$ and $M_u$, as shown by the following lemma, can be (implicitly) filled in so that both matrices are unit-Monge. 

\begin{lemma}
	Let $M$ be a triangular unit Monge matrix. The undefined values of $M$ can be implicitly filled so that M is a $n \times n$ unit Monge matrix. 
\end{lemma}
\begin{proof}
We prove the lemma for a lower triangular matrix $M$ ($n \geq i \geq j \geq 1$). The proof for an upper triangular matrix is similar. For an undefined element $M[i,j]$ (where $j > i$) we define $M[i,j] = M[j,j] - \sum_{k=i}^{j-1} (M[k+1,k] - M[k,k])$. This choice guarantees that the matrix $P$ for the completed $M$ is zero wherever an undefined element of $M$ is involved. 
Then, for $j=i+1$, 
$$
M[i+1,j] - M[i,j]  = M[i+1,i+1] - M[i+1,i+1] + M[i+1,i] - M[i,i] = 
M[i+1,i] - M[i,i],
$$
whereas for $j>i+1$ it also holds that, 
$$
M[i+1,j] - M[i,j]  = M[j,j] - \sum_{k=i+1}^{j-1} M[k+1,k] - M[k,k] - M[j,j] +  \sum_{k=i}^{j-1} M[k+1,k] - M[k,k]  = $$
$$
= M[i+1,i] - M[i,i].
$$

It follows that for all $j>i$, $M[i+1,j]-M[i,j] = M[i+1,j+1] - M[i,j+1] = M[i+1,i]- M[i,i]$, so the unit Monge property holds for $j>i$. The unit Monge property clearly holds for $j<i$, by definition of lower triangular matrix. It remains to verify that the unit Monge property holds for $j=i$. In that case
the Monge inequality $M[i+1,j]-M[i,j] \leq M[i+1,j+1] - M[i,j+1]$ holds because, substituting $j$ by $i$, we have that 
$M[i+1,i] -  M[i,i] \leq M[i+1,i+1] - M[i+1,i+1] + M[i+1,i] - M[i,i] =  M[i+1,i] - M[i,i]$. 

Note that we do not need to explicitly modify $M$ (i.e. compute the artificial values of $M$). Instead we use the fact that, by our choice of the values that we filled, an element of the matrix $P$ is zero wherever an undefined value of $M$ is involved. 
\end{proof}

We would like to use an emulator for $M_l$ and an emulator for $M_u$, but we need to eliminate paths from $r_i$ to $c_j$ for $i<j$ in $M_l$, and for $i>j$ in $M_u$. To this end we modify the construction of the emulator. For ease of presentation, we describe the construction for an $(n/2) \times n$ right substochastic matrix $P$ in which we want the distance from $c_i$ to $r_j$ to be infinite for $i>j$. It is easy to modify the construction to work for $(n/2) \times \alpha n$ matrices, for $2n \times \alpha n$ matrices, or for the case where the infinite distances are for $i<j$.
The elements of $P$ are classified into two types: an element $P[i,j]$ of $P$ is \EMPH{artificial} if $i>j$ and \EMPH{original} otherwise (that is, the $i$'th row of $P$ begins with a prefix of $i-1$ artificial elements followed by a suffix of original elements).  In the following recursive algorithm, the type of elements is always defined with respect to the full input matrix $P$ (at the top level of the recursion). 

We divide the $(n/2) \times n$ matrix $P$ into two $(n/4) \times n$ matrices $P_\uparrow$ and $P_\downarrow$.
Both $P_\uparrow$ and $P_\downarrow$ have possibly empty prefix of columns containing just artificial elements (category I), then an infix of columns containing both original and artificial elements (category II), and finally a suffix of columns containing just original elements (category III). 
Let $P'$ be the submatrix of $P_\uparrow$ induced by all the columns of category II and by the columns of category III that are not entirely zero.  Define submatrix $P''$ of $P_\downarrow$ similarly.
Note that, by definition of artificial elements (with respect to the original top-level matrix), the number of columns of category II is bounded by the number of rows in $P_\uparrow$ and $P_\downarrow$. Also, by definition of right substochastic matrix, the number of columns of category III that are not entirely zero is also bounded by the number of rows of $P_\uparrow$ and $P_\downarrow$.  
Hence, both $P'$ and $P''$ are $(n/4) \times (n/2)$ matrices for which we can recursively compute emulators $H'$ and $H''$, respectively.

We construct the emulator $H$ of $P$ from the emulators $H'$ and $H''$ as follows. 
For each column $y$ of $P$ we add a vertex $c[y]$ to $H$.  
If column $y$ in $P_\uparrow$ consists of just artificial elements (category I), we do nothing.  
If column $y$ in $P_\uparrow$ is of category II, then column $y$ is also present in $P'$, so it has a corresponding column vertex $c'[y]$ in $H'$.  Connect $c'[y]$ to $c[y]$.
If column $y$ is of category III, connect vertex $c'[y'_\to]$ of $H'$ to $c[y]$ with an edge of weight 
\[
\sum_{\substack{i > \floor{n/4} \\ j\geq y}} P[i,j]. 
\]
Treat all the columns in $P_\downarrow$ similarly, except when column $y$ is of category III, connect vertex $c''[y''_\to]$ of $H''$ to $c[y]$ with an edge of weight zero.
(For the definition of $y'_\to$ and $y''_\to$ see the proof of the bipartite case.)
An easy inductive proof shows that in this construction $r_i$ is connected to $c_j$ if and only if $i<j$.
 
The size of the emulator again satisfies the recurrence $N(n) \leq 2N(n/2) + O(n)$, and therefore is of $O(n \log n)$ overall. This concludes the proof of Lemma~\ref{L:emulator}.

\section{Distance Emulators for Planar Graphs}
\label{S:construction}

In this section we explain how to incorporate the unit-Monge emulators of Lemma~\ref{L:emulator} into the construction of planar graph emulators. The construction follows closely to the encoding scheme of Abboud \etal~\cite{agmw-nocpg-17}.  The crucial difference is that their unit-Monge matrices are represented by a compact non-graphical encoding, which we replace by the distance emulator of Lemma~\ref{L:emulator}. There are additional differences that stem from the fact that  Abboud \etal \ concentrated on the size of the encoding, and did not consider the construction time.  To achieve efficient construction time we have to modify the construction slightly.  These modifications are not directly related to the unit-Monge distance emulators.  We will describe the construction algorithm without going into details, which were treated in Abboud \etal~\cite{agmw-nocpg-17}.  Instead, we attempt to give some intuition and describe the main components of the construction and their properties.  
See Appendix~\ref{A:details} for a self-contained detailed description.

Compact representations of distances in planar graphs have been achieved in many previous works by decomposing the graph using small cycle separators and exploiting the Monge property (not the
unit-Monge property) to compactly store the distances among the vertices of the separators. The 
main obstacle in exploiting the unit-Monge property using this approach is that small cycle separators do not necessarily exist in planar graphs with face of unbounded size. 
In weighted graphs this difficulty is overcome by triangulating the graph with edges with sufficiently large weights that do not affect distances. 
In unweighted graphs, however, this does not work. 

\paragraph{Slices.}
To overcome the above difficulty, Abboud \etal \ 
showed that any unweighted planar graph contains a family of nested subgraphs (called \EMPH{components}) with some desirable properties. This nested family can be represented by a tree $\mathcal{K}$, called the \EMPH{component tree}, in which the ancestry relation corresponds to enclosure in $G$. Each component $K$ in $\mathcal K$ is a subgraph of $G$ that is enclosed by a simple cycle $\bd K$, called the \EMPH{boundary cycle} of $K$. The root component is the entire graph $G$, and its boundary cycle is the infinite face of $G$. A \EMPH{slice $K^\circ$} is defined as the subgraph of $G$ induced by all faces of $G$ that are enclosed by component $K$, and not enclosed by any descendant of $K$ in $\mathcal{K}$. The boundary cycle $\bd K$ is called the \EMPH{external hole} or the \EMPH{boundary} of $K^\circ$. The boundary cycles of the children of $K$ in $\mathcal K$ are called the \EMPH{internal holes} of $K^\circ$. Note that a slice may have many internal holes.
Abboud \etal\ showed that, for any choice of $w\in [n]$, there exists a component tree $\mathcal K$ such that the total number of vertices in the boundaries of all the slices is $O(n/w)$~\cite[Lemma~3.1]{agmw-nocpg-17}. 
For completeness, we present this proof in Appendix~\ref{SS:slice}.

\paragraph{Regions.}
Constructing distance emulator for each slice separately is not good enough na\"ively, as there could potentially be $\Omega(n/w)$ holes in a slice, and we cannot afford to store distances between each pair of holes, even if we use the unit-Monge property.
To overcome this, each slice $K^\circ$ is further decomposed recursively into smaller subgraphs called \EMPH{regions}. 
A \emph{Jordan cycle separator} is a closed curve in the plane that intersects the embedding only at $O(w)$ vertices. A slice $K^\circ$ is recursively decomposed in a process described by a binary \EMPH{decomposition tree $\mathcal R_K$}; 
each node $R$ of the tree $\mathcal R_K$ is a region of $K^\circ$, with the root being the entire slice $K^\circ$. There is a unique Jordan cycle separator $C_R$ corresponding to each internal node $R$ of $\mathcal{R}_K$.  
We abuse the terminology by referring to an internal hole of slice $K^\circ$ lying completely inside region $R$ as a \emph{hole} of $R$.
We say that a Jordan curve $C_R$ \emph{crosses} a hole $H$ of $R$ if $C_R$ intersects both inside and outside of $\bd H$. 
All Jordan cycle separators discussed here are mutually non-crossing (but may share subcurves) and each crosses a hole at most twice.

We next describe how the separators $C_R$ are chosen. Let $K$ be a component, and let $K'$ be the child component of $K$ in $\mathcal K$ with the largest number of vertices. We designate the hole $\bd K'$ of the slice $K^\circ$ as the \EMPH{heavy} hole of $K^\circ$.\footnote{This is the main difference compared with the encoding scheme of Abboud \etal\ \cite{agmw-nocpg-17}.  We will see the benefit of this technicality in Section~\ref{SS:time}.} 
For any region $R$ of $K^\circ$, let \EMPH{$R^\bullet$} denote the union of $R$ and all the subgraphs of $G$ enclosed by the internal holes of $R$ except for the heavy hole of $K^\circ$ (in other words, $R^\bullet$ is the region $R$ with all its holes except the heavy hole ``filled in''). For example, given a slice $K^\circ$, one has $(K^\circ)^\bullet = K \setminus \mathit{int}(\bd H^*)$, where $\mathit{int}(\bd H^*)$ is the interior of the heavy hole $H^*$ of $K^\circ$ (if any). 
Abboud \etal~\cite[Lemma~3.4]{agmw-nocpg-17} proved that for any region $R$ one of the following must hold (and can be applied in time linear in the size of $R$):
\begin{itemize}
\item One can find a small Jordan cycle separator $C_R$ that crosses no holes of $R$ and is balanced with respect to the number of terminals in $R^\bullet$. That is, $C_R$ encloses a constant fraction of the weight with respect to a weight function assigning each terminal in $R$ unit weight and each face corresponding to a non-heavy hole of $R$ weight equal to the number of terminals enclosed (in $R^\bullet$) by that hole. Such a separator is called a \EMPH{good} separator.
\item There exists a hole $H$ such that one can \emph{recursively} separate $R$ with small Jordan cycle separators, each of which crosses $H$ exactly twice, crosses no other hole, and is balanced with respect to the number of vertices of $\bd H$ in the region $R$. It is further guaranteed that each resulting region $R'$ along this recursion satisfies that $(R')^\bullet$ contains at most half the terminals in $R^\bullet$. The hole $H$ is called a \EMPH{disposable hole}, and the recursive process is called the \EMPH{hole elimination} process. The hole elimination process terminates if a resulting region $R'$ either contains a constant number of vertices of the disposable hole $\bd H$, or if $(R')^\bullet$ contains no terminals. Thus, the hole elimination process is represented by an internal subtree of $\mathcal R_K$ of height $O(\log n)$.
\end{itemize}

We begin with the region consisting of the entire slice $K^\circ$. If a good separator is found, we use it to separate a region $R$ into two subregions, each containing at most a constant fraction of the terminals in $R^\bullet$. If no good separator is found then we use the hole elimination process on a disposable hole $H$. We stop the decomposition as soon as a region contains at most one hole (other than the heavy hole) or at most one terminal.
Thus, the overall height of $\mathcal R_K$ is $O(\log k \cdot \log n)$. We mentioned that it is guaranteed that each Jordan cycle separator used crosses a hole at most twice. This implies that, within each region $R$, the vertices of each hole $\bd H$ of $K^\circ$ that belong to $R$ can be decomposed into $O(\log k \cdot \log n)$ intervals called \EMPH{$\bd H$-intervals}; each Jordan cycle separator increases the number of such intervals in each subregion by at most one. We refer to those intervals that belong to the boundary cycle $\bd K$ of $K^\circ$ or to the heavy hole $\bd H^*$ of $K^\circ$ as \EMPH{boundary intervals}.

\subsection{The construction}
\label{SS:construction}
We are now ready to describe the construction of the 
distance emulator.
We will build the distance emulator of $G$ from the leaves of the component tree $\mathcal{K}$ to the root (which is the whole graph $G$).
For each component $K$ in $\mathcal{K}$, we will associate an emulator that preserves the following distances.
\begin{itemize} 
\item \emph{terminal-to-terminal} distances in $K$ between any two terminals in $K$.

\item \emph{terminal-to-boundary} distances in $K$ between any terminal in $K$ and vertices of boundary $\bd K$.
\end{itemize}
We emphasize that the distances being preserved are those in $K$, not $K^\circ$.
The distance emulator for $G$ can be recovered from the emulator associates with the root of the component tree $\mathcal{K}$.

We now describe the construction of the emulator associated with component $K$, assuming all the emulators for the children of $K$ in $\mathcal{K}$ have been constructed.
The emulator of $K$ is constructed by adding the following edges and unit-Monge emulators from Lemma~\ref{L:emulator} to the emulators obtained from the children of $K$.
\begin{enumerate}  
\item  
\label{b2b}
Add unit-Monge emulators representing the distances in $K^\circ$ between all pairs of vertices on the \emph{boundary} $\bd K$ of $K$ or on the \emph{heavy hole} $\bd H^*$ of $K$.

\item For each region $R$ in the decomposition tree $\mathcal{R}_K$ of the slice $K^\circ$, add the following edges and the unit-Monge emulators according to the type of $R$.
\begin{enumerate}[(A)]
\item If $R$ is an internal node of $\mathcal{R}_K$ with separator $C_R$ (either good or hole-eliminating):

\begin{enumerate}
\item   
\label{t2sep}
Add an edge between each \emph{terminal} in $R^\bullet$ and each vertex on the \emph{separator} $C_R$. The weight of each edge is the distance in $R^\bullet$ between the two vertices.

\item   
\label{sep2b}
For each of the two subregions $R_1$ and $R_2$ of $R$, add a unit-Monge emulator representing the distance in $R_i^\bullet$ between the \emph{separator} $C_R$ and each \emph{boundary interval} of $R_i$.

\item If $R$ corresponds to a cycle separator $C_R$ eliminating a disposable hole $\bd H$:

\begin{enumerate} 
\item 
\label{h2b}
If $R$ is the root of the subtree of $\mathcal{R}_K$ representing the hole-elimination process of a \emph{hole} $H$, add a unit-Monge emulator representing the distances between $\bd H$ and each \emph{boundary interval} of $R$. The distances are in the subgraph obtained from $R^\bullet$ by removing the subgraph strictly enclosed by $\bd H$.
\item 
\label{sep2h-elim}
For each of the two subregions $R_1$ and $R_2$ of $R$, add a unit-Monge emulator representing the distance in $R_i^\bullet$ between the \emph{separator} $C_R$ and each
 \emph{$\bd H$-interval} of $R_i$. 
\end{enumerate}
\end{enumerate}

\item If $R$ is a leaf of $\mathcal{R}_K$ and all terminals in $R^\bullet$ are enclosed by a single non-heavy hole $\bd H$:

\begin{enumerate} 
\item   
\label{leaf-h2b}
Add a unit-Monge emulator representing the distances between $\bd H$ and each \emph{boundary interval} of $R$. The distances are in the subgraph obtained from $R^\bullet$ by removing the subgraph strictly enclosed by $\bd H$.
\end{enumerate}

\item If $R$ is a leaf of $\mathcal{R}_K$ that contains exactly one terminal $t$:

\begin{enumerate} 
\item    
\label{leaf-t2b}
Add an edge between the only \emph{terminal} $t$ and each vertex on 
each \emph{boundary interval} of $R$. 
The weight of each edge is the distance in $R^\bullet$ between the two vertices.

\end{enumerate}

\end{enumerate}

\end{enumerate}

The proof of correctness is essentially the same as that in Abboud \etal\ \cite[Lemma 3.6]{agmw-nocpg-17}.

\subsection{Space analysis}
\label{sec:space}

The $O(\sqrt{k\cdot n}\log^3 n)$ space analysis is essentially the same as in Abboud \etal\ \cite{agmw-nocpg-17}. The main difference in the algorithm is the introduction of the heavy hole as part of the boundary of a slice. This will allow an efficient construction time and does not affect the construction space because it increases the total size of the boundary by a constant factor.

Since a boundary cycle appears at most twice as a boundary of a slice (once as the external boundary and once as the heavy hole), the total size for all unit-Monge emulators for boundary-to-boundary distance (\ref{b2b}) is $O((n/w)\log n)$.

We next bound the total space for terminal-to-separator distances (type \ref{t2sep}) over all slices.
Whenever a terminal stores its distance to a separator, the total number of terminals in its region $R^\bullet$ decreases by a constant factor within $O(\log n)$ levels of the decomposition tree (either immediately if the separator is a good separator, or within $O(\log n)$ levels of the hole elimination process), so this can happen $O(\log^2n)$ times per terminal. Since each separator has size $O(w)$, the total space for representing distances of type \ref{t2sep} is $O(kw\log^2n)$.

The space required to store a single unit-Monge distance emulator representing distances between a separator $C_R$ and a boundary interval $b$ (type \ref{sep2b}) is $O(|C_R| + |b| \log |b|)$ (The rows of such an emulator correspond to the vertices of $b$, and the columns to vertices of $C_R$. Each row in the corresponding right sub-stochastic matrix has at most a single non-zero entry). 
We note the following facts: (i) The size of any separator $C_R$ is $O(w)$. (ii) The depth of the recursion within each slice is $O(\log^2 n)$. (iii) There are at most $k$ regions at each level of the recursion. (iv) The total number of boundary intervals in all regions at a specific recursive level is within a constant factor from the number of regions at that level, i.e., $O(k)$. Thus, the total size of all boundary intervals at a fixed level of the recursion is at most $O(k)$ plus  the size of the boundary of the slice (boundary intervals are internally disjoint, and there are $O(k)$ intervals in each region). (v) The total boundary size of all slices combined is $O(n/w)$. By the above facts, the total space for type \ref{sep2b} emulators is $O((kw + (k + n/w)\log n)\log^2 n) = O((kw + n/w)\log^3 n)$.

Hole-to-boundary distances (types \ref{h2b} and \ref{leaf-h2b}) are stored for at most one hole in each region along the recursion, and require $O(|\bd H| + |b|\log|b|)$ space. By the same arguments as above, and since the total size of hole boundaries $\bd H$ is bounded by the total size of slice boundaries, the total space for types \ref{h2b} and \ref{leaf-h2b} emulators is $O((k+n/w)\log^3 n)$.

The space for separator-to-hole emulators (type \ref{sep2h-elim}) is bounded by the separator-to-boundary distances (type \ref{sep2b}).

To bound the number of terminal-to-boundary edges (type \ref{leaf-t2b}) note that each terminal stores distances to boundary intervals that are internally disjoint for distinct terminals. Since the number of boundary terminals in a region is $O(\log n)$, the total space for type \ref{leaf-t2b} distances is $O(k\log n + n/w)$.

The overall space required is thus $O((kw + n/w)\log^3 n)$; setting $w=\sqrt{k\cdot n}$ gives the bound $O(\sqrt{k\cdot n}\log^3 n)$.

\subsection{Time analysis}
\label{SS:time}
Recall that Lemma~\ref{L:emulator} assumes that the input is an $n \times n$ unit-Monge matrix represented in $O(n\log n)$ bits as two arrays $U$ and $V$, as well as an array $P$, specifying, for each row $i$ of the right substochastic $2(n-1) \times (n-1)$ matrix $P$, the index of the (at most) single one entry in that row (see Lemma~\ref{L:substochastic}). 
Lemma~\ref{L:emulator} will be invoked on various $x\times x$ unit-Monge matrices that represent distances among a set of $x$ vertices on the boundary of a single face  of a planar graph $G$ with $n$ vertices (or sometimes, on two sets of vertices on two distinct faces).

\begin{lemma}\label{L:MSSP}
Let $G$ be a planar graph with $n$ nodes. Let $f$ be a face in $G$. One can compute the desired representation of the unit-Monge matrix $M$ representing distances among the vertices of $f$ in $O(n)$ time. 
\end{lemma}
\begin{proof}
We tweak the linear-time multiple source shortest paths (MSSP) algorithm of Eisentat and Klein~\cite{ek-ltamf-13}. Along its execution, the MSSP algorithm maintains the shortest path trees rooted at each vertex of a single distinguished face $f$, one after the other.
This is done by transforming the shortest path tree rooted at some vertex $v_i$ to the shortest path tree rooted at the next vertex $v_{i+1}$ in cyclic order along the distinguished face $f$.  Let $d_i(u)$ denote the distance from the current root $v_i$ to $v$.  The slack of an arc $uw$ with respect to root $v_i$ is defined as $\slack_i(uw) \coloneqq c(uw) + d_i(u) - d_i(w)$. 
Thus, for all arcs in the shortest path tree the slack is zero, and for all arcs not in the shortest path tree the slack is non-negative. Eisenstat and Klein show that throughout the entire execution of the algorithm there are $O(n)$ changes to the slacks of edges. Their algorithm explicitly maintains and updates the slacks of all edges as the root moves along the distinguished face. 
By definition of the matrix $P$, 
\[
P[i,j] ~=~ M[i+1,j+1] + M[i,j] - M[i+1,j] - M[i,j+1]. 
\]
Consider an edge $v_jv_{j+1}$ along the distinguished face. Its slack with respect to root $v_i$ is $1 + M[i,j] - M [i,j+1]$. Its slack with respect to root $v_{i+1}$ is $1 + M[i+1,j] - M [i+1,j+1]$. Therefore, $P[i,j] = \slack_i(v_j v_{j+1}) - \slack_{i+1}(v_j v_{j+1})$. It follows that one can keep track of the nonzero entries of $P$ by keeping track of the edges of the distinguished face whose slack changes. The entries of the arrays $U$ and $V$ can be obtained in total $O(n)$ time since they only depend on the distances from a single vertex of $f$.
\end{proof}

We can now analyze the construction time.
The boundary-to-boundary distances (type \ref{b2b}) in a single slice are computed by a constant number of invocations of MSSP (Lemma~\ref{L:MSSP}) on $K^\bullet$. We then compute the emulator for each unit-Monge matrix on $x$ vertices in $O(x \log x)$ time using Lemma~\ref{L:emulator}. 
For the total MSSP time, we need to bound the total size of the regions $R^\bullet$ on which we invoke Lemma~\ref{L:MSSP}. This is where we use the fact that the heavy hole of a slice $K^\circ$ contains at least half the vertices of $K$. Since $K^\bullet$ does not include the heavy hole of $K$, the standard heavy path argument implies that the number of slices in which a vertex $v$ participates in invocations of the MSSP algorithm is at most 1 plus $O(\log n)$ (the number of non-heavy holes $v$ belongs to). Thus, each vertex $v$ participates in invocations of MSSP for $O(\log n)$ slices. In each invocation of MSSP each vertex $v$ is charged $O(\log n)$ time, so the total time for all MSSP invocations is $O(n\log^2n)$. To bound the total time of the invocations of Lemma~\ref{L:emulator}, note that
each boundary cycle appears at most twice as a boundary of a slice (once as the external boundary and once as the heavy hole). Hence, the total time for all invocations of Lemma~\ref{L:emulator} is $O((n/w)\log n)$. Overall we get that the total time to compute all the emulators of type~\ref{b2b} is $O(n\log^2 n + (n/w)\log n) = O(n \log^2 n)$.

The terminal-to-separator distances (type \ref{t2sep}) are computed by running MSSP twice, once for the interior of the separator $C_R$ in $R^\bullet$ and once for the exterior of $C_R$ in $R^\bullet$. This takes $O(|R^\bullet| \log|R^\bullet|)$ time. 
Then, we can report the distance, within each of these two subgraphs of $R^\bullet$, from each vertex of $C_R$ to any other vertex in $O(\log n)$ time. Next, for each terminal $t$ in $R^\bullet$, we compute the distance from $t$ to all $w$ vertices of $C_R$ by running FR-Dijkstra~\cite{fr-pgnwe-06} on the graph consisting of the edges between $t$ and the vertices of $C_R$ as well as the two complete graphs representing the distances among the vertices of $C_R$ in the interior and exterior. This takes $O(w \log^2w \log n)$ time. The $ O(w\log^2 w)$ term is the running time of FR-Dijkstra and the additional $O(\log n)$ factor is because each distance accesses by FR-Dijkstra is retrieved on-the-fly from the MSSP data structure. 

We analyze separately the total cost of all the MSSP computations and the total cost of computing the edge weights. For the latter, we have argued in the analysis of the emulator's size that each terminal participates in the computation of edges of type \ref{t2sep} in $O(\log^2 n)$ regions along the entire construction. Therefore, the total time to compute all such edges is $O(kw \log^2 w \log^3 n)$. For the total MSSP time, we need to bound the total size of the regions $R^\bullet$ on which we run MSSP. This is done in a similar manner to the bound on the time for MSSP invocations for type~\ref{b2b} above, using the standard heavy path argument.
We argued that the number of slices in which a vertex $v$ participates in invocations of the MSSP algorithm is at most 1 plus $O(\log n)$ (the number of non-heavy holes $v$ belongs to). Thus, each vertex $v$ participates in invocations of MSSP for $O(\log n)$ slices. Since the depth of the decomposition tree of a slice is $O(\log^2 n)$, each vertex participates in at most $O(\log^3 n)$ invocations of MSSP in the computation of edges of type~\ref{t2sep}. 
Each vertex is charged $O(\log n)$ time in  each invocation of MSSP it participates in, so the total time for all invocations of MSSP in the process of computing edges of type \ref{t2sep} is $O(n \log^4 n)$.  

To bound the time required to compute the emulators of type~\ref{sep2b} (separator-to-boundary distances), note that the appropriate unit-Monge matrices can be computed by MSSP invocations within the same time bounds analyzed for the MSSP invocations for terminal-to-separator edges above. Hence, we only need to account for the additional time spent on constructing the emulators by invocations of Lemma~\ref{L:emulator}, which is linear in the total size of these emulators, which we have already shown above to be $O((kw+n/w)\log^3 n)$.

The analysis for the time required for the hole-to-boundary emulators (types \ref{h2b} and \ref{leaf-h2b}) uses the same arguments as the analysis for type~\ref{sep2b}. Again, the MSSP time is $O(n \log^4 n)$, and the time for invocations of Lemma~\ref{L:emulator} is linear in the size of these emulators, which is $O((k+n/w)\log^3 n)$.

The time to compute the separator-to-hole emulators (type~\ref{sep2h-elim}) is bounded by the time to compute the separator-to-boundary emulators (type~\ref{sep2b}).

To compute the terminal-to-boundary edges we again need to invoke MSSP once on $R^\bullet$ for each leaf region $R$, and then query the distance in $O(\log n)$ time per distance. The time for MSSP is dominated by the MSSP time accounted for in type~\ref{sep2b} emulators, and since we have shown that the total number of such edges is $O(k + n/w)$, the total time to compute them is $O((k +n/w)\log n)$.

The overall construction time is therefore dominated by the $O(n\log^4 n)$ term. 
Combined with the space analysis of Section~\ref{sec:space}, we establish Theorem~\ref{Th:emulator-subset}.

\paragraph*{Acknowledgment.}  The authors would like to thank Timothy Chan, Jeff Erickson, Sariel Har-Peled, and Yipu Wang for helpful discussions.  
The first author express special thanks to \begin{CJK*}{UTF8}{bsmi}施鴻逸\end{CJK*} (Hong-Yi Shih) for discussion back in NTU in the early days that made this paper possible.

\bibliographystyle{plainurl}

\begin{thebibliography}{10}

\bibitem{AbboudB16}
Amir Abboud and Greg Bodwin.
\newblock Error amplification for pairwise spanner lower bounds.
\newblock In {\em 27th SODA}, pages 841--854, 2016.

\bibitem{AbboudB17}
Amir Abboud and Greg Bodwin.
\newblock The 4/3 additive spanner exponent is tight.
\newblock {\em J. {ACM}}, 64(4):28:1--28:20, 2017.

\bibitem{AbboudBP17}
Amir Abboud, Greg Bodwin, and Seth Pettie.
\newblock A hierarchy of lower bounds for sublinear additive spanners.
\newblock In {\em 28th SODA}, pages 568--576, 2017.

\bibitem{agmw-nocpg-17}
Amir Abboud, Pawel Gawrychowski, Shay Mozes, and Oren Weimann.
\newblock Near-optimal compression for the planar graph metric.
\newblock In {\em 29th SODA}, pages 530--549, 2018.

\bibitem{SMAWK}
Alok Aggarwal, Maria~M. Klawe, Shlomo Moran, Peter~W. Shor, and Robert~E.
  Wilber.
\newblock Geometric applications of a matrix-searching algorithm.
\newblock {\em Algorithmica}, 2(1):195--208, 1987.

\bibitem{Kaplan2}
Alok Aggarwal and James~K. Park.
\newblock Notes on searching in multidimensional monotone arrays (preliminary
  version).
\newblock In {\em 29th FOCS}, pages 497--512, 1988.

\bibitem{Kaplan3}
Alok Aggarwal and Subhash Suri.
\newblock Fast algorithms for computing the largest empty rectangle.
\newblock In {\em 3rd SoCG}, pages 278--290, 1987.

\bibitem{Alon02}
Noga Alon.
\newblock Testing subgraphs in large graphs.
\newblock {\em Random Struct. Algorithms}, 21(3-4):359--370, 2002.

\bibitem{adkp-sdl-16}
Stephen Alstrup, Søren Dahlgaard, Mathias Bæk~Tejs Knudsen, and Ely Porat.
\newblock {Sublinear Distance Labeling}.
\newblock In {\em 24th ESA}, pages 5:1--5:15, 2016.

\bibitem{aghp-sfsld-16}
Stephen Alstrup, Cyril Gavoille, Esben~Bistrup Halvorsen, and Holger Petersen.
\newblock Simpler, faster and shorter labels for distances in graphs.
\newblock In {\em 27th SODA}, pages 338--350, 2016.

\bibitem{bg-sprgm-08}
Amitabh Basu and Anupam Gupta.
\newblock Steiner point removal in graph metrics.
\newblock {\em Unpublished Manuscript, available from
  http://www.math.ucdavis.edu/\~{}abasu/papers/SPR.pdf}, 2008.

\bibitem{bkmp-ncspa-05}
Surender Baswana, Telikepalli Kavitha, Kurt Mehlhorn, and Seth Pettie.
\newblock New constructions of ($\alpha$, $\beta$)-spanners and purely additive
  spanners.
\newblock In {\em 16th SODA}, pages 672--681, 2005.

\bibitem{BF10}
Guy~E. Blelloch and Arash Farzan.
\newblock Succinct representations of separable graphs.
\newblock In {\em 21st CPM}, pages 138--150, 2010.

\bibitem{Bodwin17}
Greg Bodwin.
\newblock Linear size distance preservers.
\newblock In {\em 28th SODA}, pages 600--615, 2017.

\bibitem{BodwinW16}
Greg Bodwin and Virginia~Vassilevska Williams.
\newblock Better distance preservers and additive spanners.
\newblock In {\em 27th SODA}, pages 855--872, 2016.

\bibitem{bce-sdpas-05}
B{\'e}la Bollob{\'a}s, Don Coppersmith, and Michael Elkin.
\newblock Sparse distance preservers and additive spanners.
\newblock {\em SIAM Journal on Discrete Mathematics}, 19(4):1029--1055, 2005.

\bibitem{MSMS}
Glencora Borradaile, Philip~N. Klein, Shay Mozes, Yahav Nussbaum, and Christian
  Wulff{-}Nilsen.
\newblock Multiple-source multiple-sink maximum flow in directed planar graphs
  in near-linear time.
\newblock In {\em 52nd FOCS}, pages 170--179, 2011.

\bibitem{GomoryHu}
Glencora Borradaile, Piotr Sankowski, and Christian Wulff{-}Nilsen.
\newblock Min st-cut oracle for planar graphs with near-linear preprocessing
  time.
\newblock In {\em 51st FOCS}, pages 601--610, 2010.

\bibitem{c-mdpg-12}
Sergio Cabello.
\newblock Many distances in planar graphs.
\newblock {\em Algorithmica}, 62(1-2):361--381, 2012.

\bibitem{c-sadsp-17}
Sergio Cabello.
\newblock Subquadratic algorithms for the diameter and the sum of pairwise
  distances in planar graphs.
\newblock In {\em 28th SODA}, pages 2143--2152, 2017.

\bibitem{CXKR06}
Hubert~T.{-}H. Chan, Donglin Xia, Goran Konjevod, and Andr{\'{e}}a~W. Richa.
\newblock A tight lower bound for the steiner point removal problem on trees.
\newblock In {\em 9th APPROX}, pages 70--81, 2006.

\bibitem{c-cgpgl-13}
Hsien-Chih Chang and Hsueh-I Lu.
\newblock Computing the girth of a planar graph in linear time.
\newblock {\em SIAM Journal on Computing}, 42(3):1077--1094, 2013.

\bibitem{CGH16}
Yun~Kuen Cheung, Gramoz Goranci, and Monika Henzinger.
\newblock Graph minors for preserving terminal distances approximately - lower
  and upper bounds.
\newblock In {\em 43rd ICALP}, pages 131:1--131:14, 2016.

\bibitem{ChiangLL01}
Yi-Ting Chiang, Ching-Chi Lin, and Hsueh-I Lu.
\newblock Orderly spanning trees with applications to graph encoding and graph
  drawing.
\newblock In {\em Proc. 12th Symp. Discrete Algorithms}, pages 506--515, 2001.

\bibitem{ChlamtacDKL17}
Eden Chlamt{\'{a}}c, Michael Dinitz, Guy Kortsarz, and Bundit Laekhanukit.
\newblock Approximating spanners and directed steiner forest: Upper and lower
  bounds.
\newblock In {\em 28th SODA}, pages 534--553, 2017.

\bibitem{ce-sspdp-06}
Don Coppersmith and Michael Elkin.
\newblock Sparse sourcewise and pairwise distance preservers.
\newblock {\em SIAM Journal on Discrete Mathematics}, 20(2):463--501, 2006.

\bibitem{CrochemoreLandauZiv}
Maxime Crochemore, Gad.~M. Landau, and Michal Ziv-Ukelson.
\newblock A subquadratic sequence alignment algorithm for unrestricted scoring
  matrices.
\newblock {\em SIAM Journal on Computing}, 32:1654--1673, 2003.

\bibitem{dhz-apasp-00}
Dorit Dor, Shay Halperin, and Uri Zwick.
\newblock All-pairs almost shortest paths.
\newblock {\em SIAM Journal on Computing}, 29(5):1740--1759, 2000.

\bibitem{ek-ltamf-13}
David Eisenstat and Philip~N Klein.
\newblock Linear-time algorithms for max flow and multiple-source shortest
  paths in unit-weight planar graphs.
\newblock In {\em 45th STOC}, pages 735--744, 2013.

\bibitem{EEST08}
Michael Elkin, Yuval Emek, Daniel~A Spielman, and Shang-Hua Teng.
\newblock Lower-stretch spanning trees.
\newblock {\em SIAM Journal on Computing}, 38(2):608--628, 2008.

\bibitem{ElkinP16}
Michael Elkin and Seth Pettie.
\newblock A linear-size logarithmic stretch path-reporting distance oracle for
  general graphs.
\newblock {\em {ACM} Trans. Algorithms}, 12(4):50:1--50:31, 2016.

\bibitem{EGK+14}
Matthias Englert, Anupam Gupta, Robert Krauthgamer, Harald Racke, Inbal
  Talgam-Cohen, and Kunal Talwar.
\newblock Vertex sparsifiers: New results from old techniques.
\newblock {\em SIAM Journal on Computing}, 43(4):1239--1262, 2014.

\bibitem{fr-pgnwe-06}
Jittat Fakcharoenphol and Satish Rao.
\newblock Planar graphs, negative weight edges, shortest paths, and near linear
  time.
\newblock {\em Journal of Computer and System Sciences}, 72(5):868--889, 2006.

\bibitem{fp-dtert-93}
Thomas~A. Feo and J.~Scott Provan.
\newblock Delta-wye transformations and the efficient reduction of two-terminal
  planar graphs.
\newblock {\em Oper. Res.}, 41(3):572--582, 1993.

\bibitem{gppr-dlg-04}
Cyril Gavoille, David Peleg, St{\'e}phane P{\'e}rennes, and Ran Raz.
\newblock Distance labeling in graphs.
\newblock {\em Journal of Algorithms}, 53(1):85--112, 2004.

\bibitem{gku-ssdlu-16}
Pawe{\l} Gawrychowski, Adrian Kosowski, and Przemys{\l}aw Uzna{\'n}ski.
\newblock Sublinear-space distance labeling using hubs.
\newblock In {\em 30th DISC}, pages 230--242, 2016.

\bibitem{SubmatrixMonge}
Pawe{\l} Gawrychowski, Shay Mozes, and Oren Weimann.
\newblock Submatrix maximum queries in {M}onge matrices are equivalent to
  predecessor search.
\newblock In {\em 42nd ICALP}, pages 580--592, 2015.

\bibitem{GHP17}
Gramoz Goranci, Monika Henzinger, and Pan Peng.
\newblock Improved guarantees for vertex sparsification in planar graphs.
\newblock In {\em 25th ESA}, pages 44:1--44:14, 2017.

\bibitem{GR16}
Gramoz Goranci and Harald R{\"a}cke.
\newblock Vertex sparsification in trees.
\newblock In {\em 14th WAOA}, pages 103--115, 2016.

\bibitem{gp-egsc-72}
Ronald~L Graham and Henry~O Pollak.
\newblock On embedding graphs in squashed cubes.
\newblock In {\em Graph theory and applications}, volume 303, pages 99--110.
  1972.

\bibitem{g-sptmh-01}
Anupam Gupta.
\newblock Steiner points in tree metrics don't (really) help.
\newblock In {\em 12th SODA}, pages 220--227, 2001.

\bibitem{Algorithmica13}
Danny Hermelin, Gad~M. Landau, Shir Landau, and Oren Weimann.
\newblock A unified algorithm for accelerating edit-distance via
  text-compression.
\newblock {\em Algorithmica}, 65:339--353, 2013.

\bibitem{HuangP18}
Shang{-}En Huang and Seth Pettie.
\newblock Lower bounds on sparse spanners, emulators, and diameter-reducing
  shortcuts.
\newblock In {\em 16th SWAT}, pages 26:1--26:12, 2018.

\bibitem{insw-iamcm-11}
Giuseppe~F. Italiano, Yahav Nussbaum, Piotr Sankowski, and Christian
  Wulff-Nilsen.
\newblock Improved algorithms for min cut and max flow in undirected planar
  graphs.
\newblock In {\em 43rd STOC}, pages 313--322, 2011.

\bibitem{JaJaMS04}
Joseph J{\'{a}}J{\'{a}}, Christian~Worm Mortensen, and Qingmin Shi.
\newblock Space-efficient and fast algorithms for multidimensional dominance
  reporting and counting.
\newblock In {\em 15th ISAAC}, pages 558--568, 2004.

\bibitem{knr-irg-92}
Sampath Kannan, Moni Naor, and Steven Rudich.
\newblock Implicat representation of graphs.
\newblock {\em SIAM Journal on Discrete Mathematics}, 5(4):596--603, 1992.

\bibitem{kmns-smqmm-12}
Haim Kaplan, Shay Mozes, Yahav Nussbaum, and Micha Sharir.
\newblock Submatrix maximum queries in {M}onge matrices and {M}onge partial
  matrices, and their applications.
\newblock In {\em 23rd SODA}, pages 338--355, 2012.

\bibitem{KK89}
M.~M. Klawe and D~J. Kleitman.
\newblock An almost linear time algorithm for generalized matrix searching.
\newblock {\em SIAM Journal Discrete Math.}, 3(1):81--97, 1990.

\bibitem{TALG10SP}
Philip Klein, Shay Mozes, and Oren Weimann.
\newblock Shortest paths in directed planar graphs with negative lengths: a
  linear-space ${O}(n \lg^2n)$-time algorithm.
\newblock {\em ACM Transactions on Algorithms}, 6(2):2--13, 2010.

\bibitem{KleinMS13}
Philip~N. Klein, Shay Mozes, and Christian Sommer.
\newblock Structured recursive separator decompositions for planar graphs in
  linear time.
\newblock In {\em Proceedings of the forty-fifth annual ACM symposium on Theory
  of computing}, pages 505--514, 2013.

\bibitem{kk-spqpg-03}
Lukasz Kowalik and Maciej Kurowski.
\newblock Short path queries in planar graphs in constant time.
\newblock In {\em 35th STOC}, pages 143--148, 2003.

\bibitem{KNZ14}
Robert Krauthgamer, Huy~L Nguyen, and Tamar Zondiner.
\newblock Preserving terminal distances using minors.
\newblock {\em SIAM Journal on Discrete Mathematics}, 28(1):127--141, 2014.

\bibitem{KR17arxiv}
Robert Krauthgamer and Inbal Rika.
\newblock Refined vertex sparsifiers of planar graphs.
\newblock Preprint, July 2017.

\bibitem{NussbaumOracle}
Jakub Lacki, Yahav Nussbaum, Piotr Sankowski, and Christian Wulff{-}Nilsen.
\newblock Single source - all sinks max flows in planar digraphs.
\newblock In {\em 53rd FOCS}, pages 599--608, 2012.

\bibitem{m-mug-65}
J.~W. Moon.
\newblock On minimal n-universal graphs.
\newblock {\em Glasgow Mathematical Journal}, 7(1):32--33, 1965.

\bibitem{ms-edopg-12}
Shay Mozes and Christian Sommer.
\newblock Exact distance oracles for planar graphs.
\newblock In {\em 23rd SODA}, pages 209--222, 2012.

\bibitem{MunroR97}
J.~Ian Munro and Venkatesh Raman.
\newblock Succinct representation of balanced parentheses, static trees and
  planar graphs.
\newblock In {\em Proc. 38th Annual Symposium on Foundations of Computer
  Science}, pages 118--126, 1997.

\bibitem{Schmidt1998}
Jeanette~P. Schmidt.
\newblock All highest scoring paths in weighted grid graphs and their
  application to finding all approximate repeats in strings.
\newblock {\em SIAM Journal on Computing}, 27(4):972--992, 1998.

\bibitem{tz-sesde-06}
Mikkel Thorup and Uri Zwick.
\newblock Spanners and emulators with sublinear distance errors.
\newblock In {\em 17th SODA}, pages 802--809, 2006.

\bibitem{Tiskin}
Alexander Tiskin.
\newblock Semi-local string comparison: algorithmic techniques and
  applications.
\newblock {\em Arxiv 0707.3619}, 2007.

\bibitem{TiskinSODA}
Alexander Tiskin.
\newblock Fast distance multiplication of unit-{M}onge matrices.
\newblock In {\em 21st SODA}, pages 1287--1296, 2010.

\bibitem{Turan84}
Gy{\"{o}}rgy Tur{\'{a}}n.
\newblock On the succinct representation of graphs.
\newblock {\em Discrete Applied Mathematics}, 8(3):289--294, 1984.

\bibitem{WYgirth10}
Oren Weimann and Raphael Yuster.
\newblock Computing the girth of a planar graph in ${O}(n \log n)$ time.
\newblock {\em {SIAM} J. Discrete Math.}, 24(2):609--616, 2010.

\bibitem{w-pscc-83}
Peter~M Winkler.
\newblock Proof of the squashed cube conjecture.
\newblock {\em Combinatorica}, 3(1):135--139, 1983.

\bibitem{w-lbase-06}
David~P Woodruff.
\newblock Lower bounds for additive spanners, emulators, and more.
\newblock In {\em 47th FOCS}, pages 389--398, 2006.

\bibitem{w-widpg-09}
Christian Wulff-Nilsen.
\newblock Wiener index and diameter of a planar graph in subquadratic time.
\newblock In {\em 25th EuroCG}, pages 25--28, 2009.

\bibitem{w-ctdqp-13}
Christian Wulff-Nilsen.
\newblock Constant time distance queries in planar unweighted graphs with
  subquadratic preprocessing time.
\newblock {\em Computational Geometry}, 46(7):831--838, 2013.

\end{thebibliography}

\newpage

\appendix
\section{Missing Details}
\label{A:details}

The construction follows closely to the sublinear-size encoding scheme of undirected unweighted planar graphs described by Abboud \etal~\cite{agmw-nocpg-17}, while carefully replacing each encoding of unit-Monge matrix with the corresponding distance emulator guaranteed by Lemma~\ref{L:emulator} and Section~\ref{S:monge-emulator}.
One of the main difference here is that Abboud \etal\ did not focus on or analyze the construction time.  We have introduced some changes to their construction in order to reduce the construction time of the emulator.  We will emphasize the changes in the subsequent subsections.

\subsection{Slices}
\label{SS:slice}

The \EMPH{radial graph $G^\diamond$} of a plane graph $G$ is the face-vertex incidence graph of $G$; the vertex set of $G^\diamond$ consists of the vertices and faces of $G$, and the edge set of $G^\diamond$ corresponds to all the vertex-face incidences in $G$.  The radial graph $G^\diamond$ is also a plane graph with its embedding inherited from $G$.  Notice that for any plane graph $G$ and its dual $G^*$, the radial graph $G^\diamond$ is identical.

Perform a breadth-first search from the unbounded face of $G$ in $G^\diamond$.  The \EMPH{level} of a vertex or face in $G$ is defined to be its depth in the breadth-first search tree of $G^\diamond$.  It is well-known \cite{fp-dtert-93,KleinMS13} that the odd levels (which correspond to the vertices in $G$) forms a laminar family of edge-disjoint simple cycles.  
Formally, let \EMPH{$\mathcal{K}_{\ge i}$} be the collection of connected components of subgraph of the dual graph $G^*$ induced by faces of levels at least $2i$;
for simplicity we call each connected components in $\mathcal{K}_{\ge i}$ as \EMPH{component} at level $i$.  The collection of components at all levels and their containment relationship forms a tree structure, which we refer to as the \EMPH{component tree $\mathcal{K}$}.%
\footnote{To describe the emulator construction we will be using several different trees whose nodes correspond to very different objects.  To avoid confusion, we will consistently denote trees whose nodes are subgraphs of $G$ using calligraphy letters, and denote trees whose nodes are vertices of $G$ using normal uppercase letters.}
Component $K'$ is a descendant of another component $K$ if the faces of $K'$ are a subset of the faces of $K$. The root of the component tree corresponds to the component of $G$ that is the union of all faces in $G$ except for the unbounded face.  The boundary of a component $K$ always forms a simple cycle in $G$; we often refer it as the \EMPH{boundary cycle $\bd K$} of $K$.  

A \emph{slice} of $G$ is defined to be the subgraph of $G$ induced by all faces of $G$ enclosed by the boundary cycle of some component at level $i$ and not strictly enclosed by the boundary cycles of any component at level $j$ for some $j > i$. 
In this paper we will only consider slices of the following form:  For a fixed level $\ell$ and some integer \EMPH{width} $w$, consider any slice defined by some component at the levels of the form $\ell_\alpha \coloneqq \ell + \alpha w$ for all integer $\alpha \ge 0$, as well as level $0$ (which the component corresponds to the whole graph $G$).  
For each component $K$ at level $\ell_\alpha$ in the component tree $\mathcal{K}$ one can define \EMPH{slice $K^\circ$} between component $K$ and all the descendants of $K$ in $\mathcal{K}$ at level $\ell_{\alpha+1}$.
Each slice inherits its embedding from $G$.  
By definition $K^\circ$ is a subgraph of $K$.
The boundary cycle of $K$ is also the boundary cycle of the unbounded face of the slice $K^\circ$, called the \EMPH{external hole} or \EMPH{boundary} of $K^\circ$; and the boundary cycles of components at level $\ell_{\alpha+1}$ in $K^\circ$ are called the \EMPH{internal holes} of $K^\circ$.  Bear in mind that the definition of external and internal holes depends on the slice we are currently in; an internal hole of the slice $K^\circ$ is the external hole of some slice ${K'}^\circ$ for some descendant $K'$ of $K$.

For the sake of simplicity we assume that the unbounded face of $G$ has length $3$ by adding a triangle enclosing the whole graph $G$, and add a single edge between the triangle and $G$.
Now an averaging argument, together with the length assumption on the unbounded face, guarantees that for any fixed width $w$ there will be a choice of $\ell$ such that the sum of the length of all the holes in the slices defined is at most $O(n/w)$ \cite[Lemma~3.1]{agmw-nocpg-17}. 

\subsection{Regions}
\label{SS:region}

Constructing distance emulator directly for each slice is not good enough na\"ively, as there could potentially be $\Omega(n)$ holes in a slice.
Now the strategy is to decompose the slice $K^\circ$ into smaller pieces called \emph{regions}, such that each region has ``small boundary'', has at least one terminal, and has at most a constant number of internal holes.  The distance information will be stored using the regions, such that for each slice $K^\circ$ the distances \emph{in K} between all terminals and vertices on holes are kept by the distance emulator (we will see shortly in Section~\ref{SS:construction} why it is important to store distances with respect to the whole component $K$, not just the slice $K^\circ$).

Each slice $K^\circ$ contains at most $w$ consecutive levels of breadth-first search tree.  One can triangulate the slice $K^\circ$ into $K^\circ_\Delta$ in a way that the breadth-first search tree on $K^\circ$ starting at an auxiliary vertex representing the unbounded face of $K^\circ$ is consistent with the breadth-first search tree of the radial graph $G^\diamond$ that defines the levels and slices \cite[Lemma~3.2]{agmw-nocpg-17}, and each hole of $K^\circ$ has an auxiliary vertex connecting to all the boundary vertices on the hole, except for the internal hole that contains more than half of the vertices in $K$ (if any).  (Such holes are called \EMPH{heavy}; it is straightforward to see that there can be at most one heavy hole in a slice.%
\footnote{This is the main difference to the encoding scheme of Abboud \etal\ \cite{agmw-nocpg-17}.  We will discuss this technicality in Section~\ref{SS:time}.})
More precisely, let \EMPH{$T_K$} be the breadth-first search tree of $K^\circ_\Delta$, 
then vertex $u$ is a parent of vertex $v$ in $T_K$ if and only if $u$ is a grandparent of $v$ in the breadth-first search tree of $G^\diamond$.
All the edges incident to the auxiliary vertex of a hole $\bd H$ belong to $T_K$ if $\bd H$ is external, and exactly one of such edges belongs to $T_K$ if $\bd H$ is internal.
One can compute the triangulation $K^\circ_\Delta$ and the breadth-first search tree $T_K$ from $K^\circ$ in time linear in the size of $K^\circ$.

Fix a slice $K^\circ$ for some component $K$.  We now describe a high-level process to decompose $K^\circ$ into regions using fundamental cycles with respect to $T_K$.  
Given any subgraph $R$ of the slice $K^\circ$ called a \EMPH{region} starting with $R \coloneqq K^\circ$, pick an edge $e_R$ in $K^\circ_\Delta$ but not in $T_K$.  (We will describe how to pick such an edge in Section~\ref{SS:good-disposable}.)  For each region $R$, let $C_R$ be the {\em fundamental cycle} of edge $e
_R$ with respect to $T_K$ (i.e., the cycle composed of $e
_R$ plus the unique path in $T_K$ between the endpoints of  $e_R$).
Each fundamental cycle $C_R$ will be a \emph{Jordan curve} in the plane that separates $R$ into two (potentially trivial) subgraphs that share all the vertices and edges on $C_R \cap R$.
We emphasize that by the construction of $T_K$, cycle $C_R$ only intersects each (external or internal) hole of $K^\circ$ at most twice.
Recursively decompose the two subgraphs into regions.  Then the regions of $K^\circ$ is defined to be the collection of all the regions constructed from recursion.
One can associate a \EMPH{decomposition tree $\mathcal{R}_K$} of slice $K^\circ$ that encodes the region construction process; each node of the tree $\mathcal{R_K}$ is a region of $K^\circ$, and there is a unique fundamental cycle $C_R$ corresponding to each internal node $R$ of $\mathcal{R}_K$.  
We abuse the terminology by calling an internal hole of slice $K^\circ$ lying completely inside region $R$ as a \emph{hole} of $R$.

\subsection{Good cycles and disposable holes}
\label{SS:good-disposable}

What is left is to describe the choice of fundamental cycle that corresponds to each internal node of the decomposition tree $\mathcal{R}_K$.
For any region $R$, let \EMPH{$R^\bullet$} denote the union of $R$ and all the subgraphs enclosed by the internal holes of $R$ except for the heavy one; equivalently, $R^\bullet$ is the subgraph of $K$ induced by all the faces of $K$ enclosed by the (weakly-simple) boundary cycle of $R$, after removing the interior of any hole of $R$ containing more than half of the vertices in $K$ (this technicality was discussed in detail in Section~\ref{SS:time}). 
For example, given a slice $K^\circ$, one has $(K^\circ)^\bullet = K \setminus \mathit{int}(\bd H^*)$, where $\mathit{int}(\bd H^*)$ is the interior of the heavy hole of $K^\circ$ (if any).

Fix a region $R$ in $\mathcal{R}_K$.
We say a fundamental cycle $C$ with respect to $T_K$ is \EMPH{good} for region $R$ if $C$ is a balanced separator with respect to number of terminals in $R^\bullet$,
and $C$ does not intersect the interior of any hole of $R$. 
A hole $\bd H$ of $R$ is \EMPH{disposable} if each subgraph induced on the faces enclosed by the fundamental cycle $C$ of each edge on $\bd H$ contains at most half of the terminals in $R^\bullet$. 
Abboud \etal\ proved the following structural result that either one can find a good fundamental cycle separator and recurse on the two subgraphs, or there must be a disposable hole, in which case one perform a ``hole-elimination'' process to remove the disposable hole.

\begin{lemma}[Abboud \etal\ {\cite[Lemma~3.4]{agmw-nocpg-17}}]
\label{L:good-disposable}
Given an arbitrary region $R$ such that $R^\bullet$ contains more than one terminal, there is either a good fundamental cycle for $R$ or a disposable hole in $R$.
\end{lemma}

Now we are ready to describe the construction of $\mathcal{R}_K$.
For each region $R$, if $R$ has at most one hole (other than the heavy hole of the slice $K^\circ$) and at most one terminal then we stop.
Otherwise we apply Lemma~\ref{L:good-disposable} on $R$.  If there is a good fundamental cycle $C$ for $R$ then we choose the edge defining $C$ to be the edge $e_R$ associated with $R$.  Let $R'$ and $R''$ be the two subgraphs induced by the faces of the interior and exterior of $C$.  Attach $R$ with two children $R'$ and $R''$ in the  decomposition tree $\mathcal{R}_K$, and recursively construct the subtrees rooted at $R'$ and $R''$.
If there is a disposable hole $\bd H$, we cut the region $R$ open along each path in $T_K$ from the root to the vertices of $\bd H$.  Attach a binary tree of height $O(\log n)$ rooted at $R$ in $\mathcal{R}_K$ by iteratively finding balanced fundamental cycle separators with respect to the number of vertices on $\bd H$ among those using the artificial edges incident to the auxiliary vertex in the hole $\bd H$ and decompose the region $R$ into two smaller pieces.  The leaves of the attached tree correspond to the regions obtained from cutting the paths in $T_K$ from root to all the vertices of $\bd H$.  Remove all the regions that do not contain any terminals; recursively construct the subtrees rooted at the rest of the regions. 

By definition a good fundamental cycle separator decreases the number of terminals in the resulting regions by a constant factor.
Similarly, by definition of disposable holes, each of the regions obtained at the end of the hole-elimination process has at most half of the terminals in the original region.
Consequently the height of the decomposition tree $\mathcal{R}_K$ is at most $O(\log |T| \cdot \log n)$ for each slice $K^\circ$.

As we mentioned before, while the boundary of a slice $K^\circ$ must be a simple cycle, that's not the case for the regions in $\mathcal{R}_K$ in general; additionally, the boundary of a region consists of fragments of vertices coming from either the holes (both external and internal) or some fundamental cycle separators.  Each fragment that comes from some hole $\bd H$ must be a simple path; we call such fragments \EMPH{$\bd H$-intervals}.  We refer those intervals that belong to the boundary cycle $\bd K$ of $K$ or to the heavy hole $\bd H^*$ of $K^\circ$ as \EMPH{boundary intervals}. Clearly, the region $K^\circ$ has at most two boundary intervals (one is the simple external boundary of $K$ and the other is the simple boundary of the heavy hole of $K^\circ$). By construction of the spanning tree $T_k$, whenever a region is separated by a fundamental cycle separator with respect to $T_k$, the number of boundary intervals increases by at most two. Using the analysis on the recursion height of the decomposition tree $\mathcal{R}_K$, we conclude that for each hole $\bd H$, there are at most $O(\log |T| \cdot \log n)$ $\bd H$-intervals on the boundary of $R$. 
We will use this fact in Section~\ref{SS:construction} to bound the size of the construction.  
\end{document}